\documentclass[a4paper]{amsart}

\usepackage[utf8]{inputenc}
\usepackage[english]{babel}
\usepackage{amsmath}
\usepackage{amsfonts}
\usepackage{amsthm}
\usepackage{amssymb}
\usepackage{mathrsfs}			
\usepackage[shortlabels]{enumitem}
\usepackage{tikz}						
\usetikzlibrary{matrix,arrows,arrows.meta}	

\usepackage{todonotes}

\pdfoutput=1

\allowdisplaybreaks		

\setlist[2]{labelsep = .5\parindent}
\setlist[3]{labelsep = .5\parindent}

\newcommand{\IR}{\mathbb{R}}

\newcommand{\fg}{\mathfrak{g}}

\newcommand{\fU}{\mathfrak{U}}

\newcommand{\cH}{\mathcal{H}}
\newcommand{\cG}{\mathcal{G}}

\newcommand{\sfG}{\mathsf{G}}
\newcommand{\sfR}{\mathsf{R}} 

\newcommand{\set}[2]{\{#1 \colon #2 \}}
\newcommand{\inp}[2]{\langle #1, #2 \rangle}

\newcommand{\Id}{\text{\normalfont Id}}
\newcommand{\red}{\text{\normalfont red}}

\newcommand{\dom}{\text{\normalfont Dom}}
\newcommand{\Graph}{\text{\normalfont Graph}}

\usepackage[all]{xy}

\usepackage{color}

\newcommand{\Cs}{C$^{\ast}$-}

\renewcommand{\d}{\mathrm{d}}

\newcommand{\rmA}{\mathrm{A}}

\newcommand{\spc}[1]{ \mathcal K_{#1}} 
\newcommand{\A}{A} 
\newcommand{\R}{R} 
\newcommand{\Rp}{R^{\prime}} 
\newcommand{\Rr}{R^{\red}} 
\newcommand{\sfRr}{\sfR^{\red}}

\newenvironment{myenum}
{\begin{enumerate}[{(}1{)}]}
{\end{enumerate}}

\theoremstyle{plain}
\newtheorem{thrm}{Theorem}

\newtheorem{lem}[thrm]{Lemma}
\newtheorem{prop}[thrm]{Proposition}
\newtheorem{cor}[thrm]{Corollary}

\theoremstyle{definition}
\newtheorem{defi}[thrm]{Definition}
\newtheorem{ass}[thrm]{Assumption}

\newtheorem{rem}[thrm]{Remark}

\title{Quantum lattice gauge fields and groupoid \Cs algebras}
\author{Francesca Arici}
\author{Ruben Stienstra}
\author{Walter D. van Suijlekom}
\address{Institute for Mathematics, Astrophysics and Particle Physics, Radboud University Nijmegen, Heyendaalseweg 135, 6525 AJ Nijmegen, The Netherlands}
\email{f.arici@math.ru.nl}
\email{R.Stienstra@math.ru.nl} 
\email{waltervs@math.ru.nl}

\date{\today}

\begin{document}


\begin{abstract}
We present an operator-algebraic approach to the quantization and reduction of lattice field theories. Our approach uses groupoid \Cs algebras to describe the observables and exploits Rieffel induction to implement the quantum gauge symmetries. We introduce direct systems of Hilbert spaces and direct systems of (observable) \Cs algebras, and, dually, corresponding inverse systems of configuration spaces and (pair) groupoids. The continuum and thermodynamic limit of the theory can then be described by taking the corresponding limits, thereby keeping the duality between the Hilbert space and observable \Cs algebra on the one hand, and the configuration space and the pair groupoid on the other. Since all constructions are equivariant with respect to the gauge group, the reduction procedure applies in the limit as well.

\end{abstract}

\maketitle

\section{Introduction}

\noindent
Yang--Mills gauge theories were introduced to model fundamental physical forces such as the weak and strong interactions. Mathematically speaking, a classical gauge theory corresponds to connections on a principal fibre bundle over spacetime, whose structure group is a Lie group. Despite their great success in physics, quantizing such theories in a mathematically rigorous way turns out to be an extremely difficult problem.

In the 1970s, Wilson tried to simplify the problem by replacing spacetime with a finite 4D lattice \cite{wilson74}.
Since the number of points is now finite, one can rigorously define path integrals.
Moreover, the number of degrees of freedom related to the connection or gauge field is now a multiple of the number of edges of the lattice, and the finiteness of lattices also suppresses IR and UV divergences.
Wilson then tried to reconstruct the continuum theory by letting the lattice spacing tend to zero.
Soon after, Kogut and Susskind put Wilson's lattice gauge theories in the framework of Hamiltonian mechanics \cite{kogut75}, choosing a Cauchy surface in spacetime and replacing it with a 3D lattice, whilst retaining time as a continuous variable.

This formulation is particularly appealing, because quantization of Hamiltonian systems has been studied extensively in the mathematical physics literature. 
In addition, there is a well-developed theory of symmetries of such systems, initiated by Dirac in \cite{dirac50} and put into the language of symplectic manifolds by Arnold and Smale.
Marsden and Weinstein studied reduction of symplectic manifolds with respect to an equivariant moment map \cite{marsden74}, which allows one to remove gauge symmetries present in gauge theories in a systematic way.
Reduction of the corresponding quantum systems can be carried out by means of an induction procedure due to Rieffel \cite{rieffel74} (cf. \cite[IV.2]{landsman98}). 

In this paper, we give a novel operator algebraic approach to the quantization of Hamiltonian lattice gauge theories using groupoid \Cs algebras. We relate this to the work of Kijowski and Rudolph on finite-lattice approximation of Hamiltonain QCD \cite{KR02,kijowski04}. We discuss how gauge theories corresponding to `finer' lattices, or more generally, to graphs, are related to coarser ones. At the Hilbert space level this was described mathematically by Baez in \cite{baez96}, whose results we extend to the obervable algebras and Hamiltonians as well. In particular, we identify the groupoid describing the limit observable algebra as well. Hence this provides a framework for constructing the infinite volume and continuum limits of such theories. The first, also called the thermodynamic limit, has recently been studied along similar lines on a lattice in \cite{grundling13,GR15}, but without the groupoid description. On the other hand, it should be noted that we restrict ourselves to `pure gauge theories', i.e.\ we do not consider the interaction of gauge fields with matter fields, and that we only consider the free, `electric' part of such fields in our Hamiltonians because of their nice behaviour in passing to finer lattices.
The study of the system with interactions is much more involved, and subject of future research.

This paper is organized as follows.
In Section \ref{sec:class_quantum_sys} we review the classical Hamiltonian lattice gauge theory, and its quantum mechanical counterpart.
In Section \ref{sec:refinements}, we recall some old results and develop some new methods to relate lattices with different lattice spacings, as well as the corresponding classical and quantum systems, which is necessary for constructing the thermodynamic and continuum limits.
We also describe the behaviour of groupoid \Cs algebras associated to refinements of graphs.
In Section \ref{sec:limit} we describe the behaviour of the system with respect to the continuum limit and also identify the groupoid that describes the continuum limit.
We finish the paper with an outlook on the dynamics of the continuum limit observable algebra.

\subsection*{Acknowlegements}	 We would like to thank Johannes Aastrup, Sergio Doplicher, Klaas Landsman, Bram Mesland, Ralf Meyer, Jean Renault, Adam Rennie and Alexander Stottmeister for useful comments and helpful discussions. This research was partially supported by NWO under the VIDI-grant \mbox{016.133.326}. 

\section{The classical and quantum system}
\label{sec:class_quantum_sys}

\noindent
We first give a brief review of classical lattice gauge theories. Then we present the mathematical setup for the corresponding quantum system, along the lines of strict deformation quantization (cf. \cite{rieffel98} and \cite{landsman98}). We also describe the observable algebras as groupoid \Cs algebras and discuss reduction of the quantum system.

\subsection{The classical system}

Let $\pi \colon P \rightarrow M$ be a left principal fibre bundle over a smooth manifold $M$ with compact structure group $G$. Thus we shall assume throughout the rest of the text that (somewhat unconventionally) $G$ acts on $P$ from the left, and that, given a local trivialization $\Phi \colon \pi^{-1}(U) \rightarrow U \times G$ and a point $p \in \pi^{-1}(U)$ with $\Phi(p) = (x,a)$, we have
\begin{equation*}
\Phi(g \cdot p) = (x,g \cdot a), 
\end{equation*}
for each $g \in G$. 

Clearly, a change of local trivializations amounts to multiplying the element $a \in G$ in $\Phi(p) = (x,a)$ by some element $g_x \in G$; such a transformation is called a {\em local gauge transformation}. 

Now suppose that $P$ carries a connection, and that $\gamma$ is some piecewise smooth path in $M$ from a point $x$ to another point $y$ in $M$. The connection then induces a $G$-equivariant diffeomorphism $f_\gamma \colon \pi^{-1}(\{x\}) \rightarrow \pi^{-1}(\{y\})$ through parallel transport along $\gamma$. In terms of local trivializations around $x$ and $y$ we find by $G$-equivariance that $f_\gamma$ is simply given by multiplication in the fiber by an element $a \in G$. Now, if we apply a change of local trivializations around $x$ and $y$ to the parallel transporter $f_\gamma$, we find that local gauge transformations act on $f_\gamma$ by sending $a \mapsto g_x a g_y^{-1}$ (in the above notation).

This transformation rule for parallel transporters is the starting point of lattice gauge theory.
Indeed, we assume that $M$ is of the form $M^\prime \times \IR$, where $M^\prime$ is a Cauchy surface, or, more generally, a hypersurface in $M$.
Next, we replace the manifold $M^\prime$ by a finite subset $\Lambda^0$ of $M^\prime$, and restrict the principal fibre bundle to this set by working in the temporal gauge, thus killing the temporal component of the gauge field, and considering the set $P|_{\Lambda^0} := \pi^{-1}(\Lambda^0)$ as the (total space of the) new principal fibre bundle.
Subsequently we choose a finite set of paths $\Lambda^1$ in $M^\prime$ between points in $\Lambda^0$, which comes with two maps $s,t \colon \Lambda^1 \rightarrow \Lambda^0$, the {\em source} and {\em target maps}, that assign to a path its starting point and end point, respectively.
The pair $\Lambda := (\Lambda^0, \Lambda^1)$ then becomes a finite oriented graph, and accordingly refer to elements of $\Lambda^1$ as {\em edges} in the rest of the paper.

\begin{ass}
\label{ass:graphs}
We require that between any two vertices in $\Lambda$ there exists at most one edge, that no edge has the same source and target (i.e. the corresponding path is a loop), and that $\Lambda$ is connected when viewed as an unoriented graph.
\end{ass}

\noindent
The bundle $P|_{\Lambda^0}$ has a discrete base space, so it is trivialisable.
A choice of a trivialization yields an identification of $P|_{\Lambda^0}$ with $G^{\Lambda^0}$, the space of functions from $\Lambda^0$ to $G$.
Connections on $P$ are now given by the induced parallel transporters associated to the elements of $\Lambda^1$.
Having chosen a local trivialization of $P|_{\Lambda^0}$, we may identify each of these parallel transporters with an element of $G$, as we explained above; in that way, the space of connections is simply identified with $G^{\Lambda^1}$. Since we are interested in studying these approximations, the compact Lie group 
$\spc{}:=G^{\Lambda^1}$ will be the configuration space of our system.

\subsubsection{Gauge symmetries}
A gauge transformation now corresponds to an element $(g_x)_{x \in \Lambda^0} \in G^{\Lambda^0}$; we denote the gauge group consisting of such elements by $\cG$.
From the transformation rule derived above on parallel transporters, we see that there is an action of $\cG$ on $\spc{}$ given by 
\begin{equation}
G^{\Lambda^0} \times G^{\Lambda^1} \rightarrow G^{\Lambda^1}, \quad 
((g_x)_{x \in \Lambda^0}, (a_e)_{e \in \Lambda^1}) \mapsto (g_{s(e)} a_e g_{t(e)}^{-1})_{e \in \Lambda^1}.
\label{eq:action_on_configuration_space}
\end{equation}
This action of the gauge group on the configuration space extends naturally to an action of $\cG$ on the phase space $T^\ast \spc{} \simeq (T^\ast G )^{\Lambda^1}$.
Explicitly, the action is given by
\begin{align*}
\cG \times (T^\ast G )^{\Lambda^1} & \rightarrow (T^\ast G)^{\Lambda^1}, \\ 
((g_x)_{x \in \Lambda^0}, (a_e, \xi_e)_{e \in \Lambda^1}) & \mapsto  (g_{s(e)} a_e g_{t(e)}^{-1}, \xi_e \circ (T_{a_e} (L_{g_{s(e)}} \circ R_{g_{t(e)}^{-1}}))^{-1})_{e \in \Lambda^1}.
\end{align*}

\begin{rem}
The above action of $\cG$ on $T^*\spc{}$ preserves the canonical symplectic form and there is a canonical momentum map for this phase space.
However, since the action of the gauge group on the configuration and/or phase space is not free, the associated Marsden--Weinstein quotient is not a manifold.
The analysis of the reduced phase space in a simple example of a lattice consisting of one plaquette can be found in \cite{FRS07,Hue11,HRS09}.
The analysis of the reduced phase space for the general case can be done along the same lines using spanning trees in the graph $\Lambda$, but this goes beyond the scope of the present paper. 
\end{rem}

\subsubsection{The classical Hamiltonian}
We may draw an analogy between the above system and a collection of spherical rigid rotors, where each rotor sits on one of the links (which in fact arises for the special case $G = \text{\normalfont SO}(3)$).
We then find that the free Hamiltonian of the system is given by 
\begin{equation}
H((a_e,\xi_e)_{e \in \Lambda^1}) = \sum_{e \in \Lambda^1} \frac{1}{2} I_e \xi_e^2,
\label{eq:classical_Hamiltonian}
\end{equation}
where $I_e$ denotes the `moment of inertia' 
associated to the link $e$, and $\xi_e$ is its angular velocity (seen as elements of the Lie algebra of $G$). This free Hamiltonian describes the `electric part' since  $\xi_e$ is proportional to the time derivative of the gauge field at the link $e$.
A full description of the gauge system should incorporate additional terms in the Hamiltonian that correspond to the `magnetic part'. These terms are gauge-invariant quantities that depend on the gauge field $(a_e)_{e \in \Lambda^1}$, such as traces of Wilson loops. We refer to \cite{kogut75} for a more extensive discussion.

\subsection{The quantum system}
Next, we discuss the quantization of the canonical system $T^*(\spc{}) = T^*(G^{\Lambda^1})$  defined in the previous section, adopting the \Cs algebraic approach to quantization of the cotangent bundle as described in \cite[Section II.3]{landsman98}. In line with Weyl quantization of $T^* \mathbb R^n$, the quantization of $T^* Q$ for any compact Riemannian manifold $Q$ is given there by the observable algebra $B_0(L^2(Q))$, the space of compact operators on $L^2(Q)$. Since the compact Lie group $\spc{}$ is naturally a compact Riemannian manifold, we find that the quantized observable algebra of $T^*(\spc{})$ is given by $A:= B_0(L^2(\spc{}))$ and the Hilbert space is $\cH= L^2(\spc{})$. Note that this is in line with the finite-lattice approximation of Hamiltonian QCD in \cite{KR02,kijowski04}.

Geometrically, we can also realize this \Cs algebra as a groupoid \Cs algebra. The construction is based on the pair groupoid $\sfG = \spc{} \times \spc{}$ so we first recall its general definition (cf. \cite[Section 3]{buneci}). 

\begin{defi}
Let $X$ be a locally compact Hausdorff space. The {\em pair groupoid} associated to $X$ has object space $X$ and space of morphism $X \times X$, with source and target maps given by the projections onto the first and the second factor respectively. Composition of morphism is given by concatenation and the inverse by $(x,y)^{-1} = (y,x)$. Note that all free and transitive groupoids are necessarily pair groupoids.
\end{defi}

\noindent
Now suppose that $X$ is endowed with a Radon measure $\mu$ of full support $X$. Recall that this is a measure on the Borel $\sigma$-algebra of $X$ that is locally finite and inner regular.
The $\ast$-algebra $C_c(X \times X)$, with convolution product 
\begin{equation*}
(\phi_1 \ast \phi_2)(x_1,x_2) := \int_X \phi_1(y, x_2) \phi_2(x_1, y) \: \d\mu(y)
\end{equation*}
and involution $\phi^*(x_1,x_2):= \overline{\phi(x_2,x_1)}$,  
is then represented by compact operators on $L^2(X, \mu)$.
Indeed, given $h \in C_c(X \times X)$, the associated integral operator $T_h$ on $L^2(X)$ is given by
\begin{equation}
\label{eq:integral-operator}
T_h \psi(x) := \int_X h(y,x) \psi(y) \: \d\mu(y).
\end{equation}
By definition the {\em reduced groupoid \Cs algebra} $C^*_{\mathrm{r}}(X \times X)$ is the closure in $B(L^2(X))$ of the image of the above representation. This is actually isomorphic to the \emph{full} groupoid \Cs algebra and one has
\begin{equation}
\label{eq:groupoid-C*-compacts}
C^*_{\mathrm{r}}(X \times X) \simeq C^*(X \times X) \simeq B_0(L^2(X)).
\end{equation}
\noindent 
We refer to \cite{renault80} for full details on the construction of groupoid \Cs algebras, see also \cite[III.3.4 and III.3.6]{landsman98}. The relation of this construction to strict quantization can be found in \cite[III.3.12]{landsman98}.

If we specialise to our case for which $\spc{}= G^{\Lambda^1}$ is our configuration space, this leads us to consider the pair groupoid $\sfG:= \spc{} \times \spc{}$, whose space of morphism is $\sfG^{(1)} = \spc{} \times \spc{}$ and whose space of objects is $\sfG^{(0)} =\spc{}$.
Thus the observable algebra $A$ is isomorphic to $C^*(\sfG)$.

\subsubsection{Gauge symmetries and reduction of the quantized system}
We will now discuss the reduction of the quantum system with respect to the gauge group. 

For this we use a procedure known as Rieffel induction \cite{rieffel74}, which in this context can be considered as the quantum analogue of Marsden-Weinstein reduction \cite{landsman95}.

Clearly, there is a unitary representation $U$ of the gauge group $\cG= G^{\Lambda^0}$ on $\cH = L^2(\spc{})$, which is given by dualizing the action in Equation \eqref{eq:action_on_configuration_space}:

\begin{equation}
\label{eq:quantum_gauge_group_action}
U(  (g_x)_{x\in \Lambda^0}) \psi ((a_e)_{e \in \Lambda^1} ) = \psi\left( (g_{s(e)} a_e g_{t(e)}^{-1})_{e \in \Lambda^1}  \right),
\end{equation}
for all $\psi \in \cH$.
Because $\cG$ is compact, we can use this representation to endow $\cH$ with the structure of a right Hilbert $C^\ast(\cG)$-module $\mathcal E$, where $C^\ast(\cG)$ denotes the group \Cs algebra of $\cG$.
Indeed, a $C^\ast(\cG)$-valued inner product is determined by
\begin{equation*}
\langle \phi, \psi \rangle (g) = \langle \psi, U ((g_x)_{x\in \Lambda^0}) \psi\rangle_\cH,
\end{equation*}
thus defining an element in $C(\cG) \subset C^\ast(\cG)$.
The reduced Hilbert space is then given by the balanced tensor product
\begin{equation*}
\mathcal E  \otimes_{C^\ast(\cG)} V,
\end{equation*}
with a representation module $V$ of $\cG$.
In our case of interest $V$ is the trivial representation, and one arrives at the Hilbert space $\cH^{\cG}$
of $\cG$-invariant vectors in $\cH$.

\begin{rem}
Since the group $\cG$ is compact, the quotient space $\cG \backslash \spc{}$ is also a compact Hausdorff space. Using the Riesz--Markov representation theorem, which establishes a one-to-one correspondence between positive functionals on $C_0(X)$ and Radon measures, we can push forward the measure $\mu$ to a measure $\mu_{\cG}$, to obtain the Hilbert space $L^2(\cG \backslash \spc{}, \mu_{\cG})$. This Hilbert space is naturally isomorphic to $\cH^{\cG}$. 
\end{rem}

\noindent
At the level of the observable algebra, one first considers the algebra $\A^\cG$ of elements of $\A$ that commute with the unitary representation of $\cG$. Since the space $\cH^\cG$ is invariant under these observables, we obtain a representation $\pi$ of the \Cs algebra $\A^\cG$ on $\cH^\cG$. It can be shown that this representation is not necessarily faithful but that $\A^\cG/\ker(\pi) \cong B_0(\cH^\cG)$ and that the generators of $\ker (\pi)$ implement the local Gauss law \cite{KR02,StvS}.
In view of equation \eqref{eq:groupoid-C*-compacts}, it is natural to associate the pair groupoid $(\cG \backslash \spc{}) \times (\cG \backslash \spc{})$ to the reduced system.

\subsubsection{The quantum Hamiltonian}
The Hamiltonian for quantum lattice gauge fields was introduced by Kogut and Susskind in \cite{kogut75}.
In analogy with the classical Hamiltonian of Equation \eqref{eq:classical_Hamiltonian}, its free part is given by the following differential operator on $C^\infty(G^{\Lambda^1})$:
\begin{equation}
\label{eq:hamiltonian}
H_0 = \sum_{e \in \Lambda^1} -\frac 12 I_e \Delta_e
\end{equation}
where $\Delta_e$ is the Laplacian on $G$, or, which is the same, the quadratic Casimir of $G$.
The operator $H_0$ is essentially self-adjoint on $C^\infty(G^{\Lambda^1}) \subset L^2(G^{\Lambda^1})$; we let $H_0$ denote its closure with domain $\dom(H_0) \subset L^2(G^{\Lambda^1})$. 
Since $H_0$ is the differential operator associated to the quadratic Casimir element of $G$, it is well-behaved with respect to the action of the gauge group:

\begin{prop}
Let $\cH := L^2(\spc{})$.
The operator $H_0$ is equivariant with respect to the action of the gauge group defined in equation \eqref{eq:quantum_gauge_group_action}.
Its restriction $H_{0, \red}$ to $\dom(H_0) \cap \cH^\cG$ is a self-adjoint operator on $\cH^\cG$, and the following diagram
\begin{equation*}
\xymatrix{
\dom(H_0) \ar[r]^-{H_0} \ar[d]_{p_{\cH^\cG}} & \cH \ar[d]^{p_{\cH^\cG}} \\
\dom(H_0) \cap \cH^\cG \ar[r]^-{H_{0, \red}} & \cH^\cG
}
\end{equation*}
is commutative.
\label{prop:reduction_of_the_electric_part_of_the_Hamiltonian}
\end{prop}

\begin{proof}
Note that $H_0|_{C^\infty(\sfG^{(0)})}$ is equivariant with respect to the left-regular representation since it is a left-invariant differential operator.
It is also equivariant with respect to the right regular representation, since $\Omega_e$ lies in the center of $\fU(\fg^{\Lambda^1})$ for each $e \in \Lambda^1$.
Thus $H_0|_{C^\infty(\sfG^{(0)})}$ is equivariant with respect to the action of the product of the two aforementioned representations, so in particular, it is equivariant with respect to the action of $\cG$.

An immediate consequence of this equivariance is that $H_0|_{C^\infty(\sfG^{(0)})}$ leaves $\cH^\cG$ invariant.
Because $H_0$ is by definition the closure of $H_0|_{C^\infty(\sfG^{(0)})}$, the space $\cH^\cG$ is also an invariant subspace for $H_0$.
In addition, since the orthogonal projection $p_{\cH^\cG}$ onto $\cH^\cG$ is a strong limit of linear combinations of unitary operators associated to elements of $H$, we have $p_{\cH^\cG}(\dom(H_0)) \subseteq \dom(H_0) \cap \cH^\cG$, and the above diagram is indeed commutative.

Finally, we prove that $H_{0, \red}$ is self-adjoint.
Let $J \colon \cH^2 \rightarrow \cH^2$ be the operator given by $(x,y) \mapsto (-y,x)$.
Then we have $\cH^2 = \Graph(H_0) \oplus J(\Graph(H_0))$ by self-adjointness of $H_0$.
From the fact that $\Graph(H_{0, \red}) \subseteq \Graph(H_0)$, we infer that $\Graph(H_{0, \red}) \perp J(\Graph(H_{0, \red}))$.
On the other hand, it follows from our discussion in the previous paragraph that
\begin{equation*}
\Graph(H_{0, \red})
= \set{(p_{\cH^\cG}(x),p_{\cH^\cG}(y))}{(x,y) \in \Graph(H_0)},
\end{equation*}
so $\Graph(H_{0, \red}) + J(\Graph(H_{0, \red})) = (\cH^\cG)^2$, hence
\begin{equation*}
\Graph(H_{0, \red}) \oplus J(\Graph(H_{0, \red}))
= (\cH^\cG)^2,
\end{equation*}
which shows that $\Graph(H_{0, \red})$ is indeed self-adjoint.
\end{proof}

\section{Refinements of the quantum system}
\label{sec:refinements}

\noindent
Our approach towards formulating a continuum limit from a gauge theory on a graph is based on a suitable notion of embeddings of graphs, referred to as `refinements'. 

We follow Baez \cite{baez96} in his description of an inverse system of configuration spaces and a direct system of Hilbert spaces, both indexed over the set of graphs with partial order given by refinement.
After reviewing this construction, we will extend this description to the level of the pair groupoids, the corresponding observable \Cs algebras and the (free) Hamiltonians. 
\subsection{Refinements of graphs}

We start by recalling the following notion (cf. {\cite[Theorem II.7.1]{maclane}}):
\begin{defi}
Let $\Lambda = (\Lambda^0, \Lambda^1)$ be an oriented graph satisfying Assumption \ref{ass:graphs}. The {\em free} or {\em path category generated by $\Lambda$}, denoted by $\mathsf{C}_\Lambda$, is defined as follows:
\begin{itemize}
\item Its set of objects is $\Lambda^0$;
\item Let $x,y \in \Lambda^0$.
The set of morphisms from $x$ to $y$ is given by the collection of orientation respecting paths in $\Lambda$ with starting point $x$ and end point $y$;
\item Composition of morphisms is given by concatenation of paths;
\item The identity element of each object $x \in \Lambda^0$ is the path of length $0$ starting and ending at $x$.
\end{itemize}
\end{defi}

\noindent
This predicates the following formulation of embedding a graph into another one:

\begin{defi}
\label{def:refinements}
Let $\Lambda_i$ and $\Lambda_j$ be two oriented graphs with corresponding free categories $\mathsf{C}_{\Lambda_i}$ and $\mathsf{C}_{\Lambda_j}$.
Suppose in addition that there exists a functor $\iota_{i,j} \colon \mathsf{C}_{\Lambda_i} \rightarrow \mathsf{C}_{\Lambda_j}$ such that:
\begin{myenum}
\item The map $\iota^{(0)}_{i,j} \colon \Lambda_i^0 \rightarrow \Lambda_j^0$ between the sets of objects is an injection;
\item The map $\iota^{(1)}_{i,j}$ between the sets of morphisms maps elements of $\Lambda_i^1$ (identified with their corresponding paths) to paths in $\Lambda_j$ such that
\begin{itemize}
\item Each edge $e \in \Lambda_i^1$ is mapped to a nontrivial path under the map $\iota^{(1)}_{i,j}$;
\item If $e$ and $e^\prime$ are distinct elements of $\Lambda^1$, then $\iota^{(1)}_{i,j}(e)$ and $\iota^{(1)}_{i,j}(e^\prime)$ have no common edges.
\end{itemize}
\end{myenum}
We call the triple $(\Lambda_i, \Lambda_j, \iota_{i,j})$ a {\em refinement} of the graph $\Lambda_i$. 
Given such a refinement, we say that $\Lambda_i$ is {\em coarser} than $\Lambda_j$, and that $\Lambda_j$ is {\em finer} than $\Lambda_i$.

When no confusion arises, we will omit the subscript $_{i,j}$ from $\iota$.
\end{defi}

\begin{rem}
Given three graphs $\Lambda_i$, $\Lambda_j$ and $\Lambda_k$, and refinements $(\Lambda_i, \Lambda_j, \iota_{i,j})$ and $(\Lambda_j, \Lambda_k, \iota_{j,k})$, then there exists a canonical refinement $(\Lambda_i, \Lambda_k, \iota_{i,k})$, where we have
$\iota_{i,k} = \iota_{j,k} \circ \iota_{i,j}$.

\label{rem:successive_refinement1}
\end{rem}

\noindent 
This allows us to define another category:
\begin{defi}
We let $\mathsf{Refine}$ denote the category with the following properties:
\begin{itemize}
\item Its set of objects is the class of oriented graphs;
\item Given two oriented graphs $\Lambda_i$ and $\Lambda_j$, then the set of morphisms from $\Lambda_i$ to $\Lambda_j$ is given by the set of refinements $(\Lambda_i, \Lambda_j, \iota)$.
\item Composition is given by composition of refinement functors.
\item For each oriented graph $\Lambda$, there is a canonical refinement $(\Lambda, \Lambda, \Id)$, where $\Id^{(0)}$ and $\Id^{(1)}$ are the identity maps on the spaces of objects and morphisms in $\mathsf{C}_\Lambda$.
\end{itemize}
\end{defi}

\begin{figure}
\includegraphics[width=0.9\textwidth]{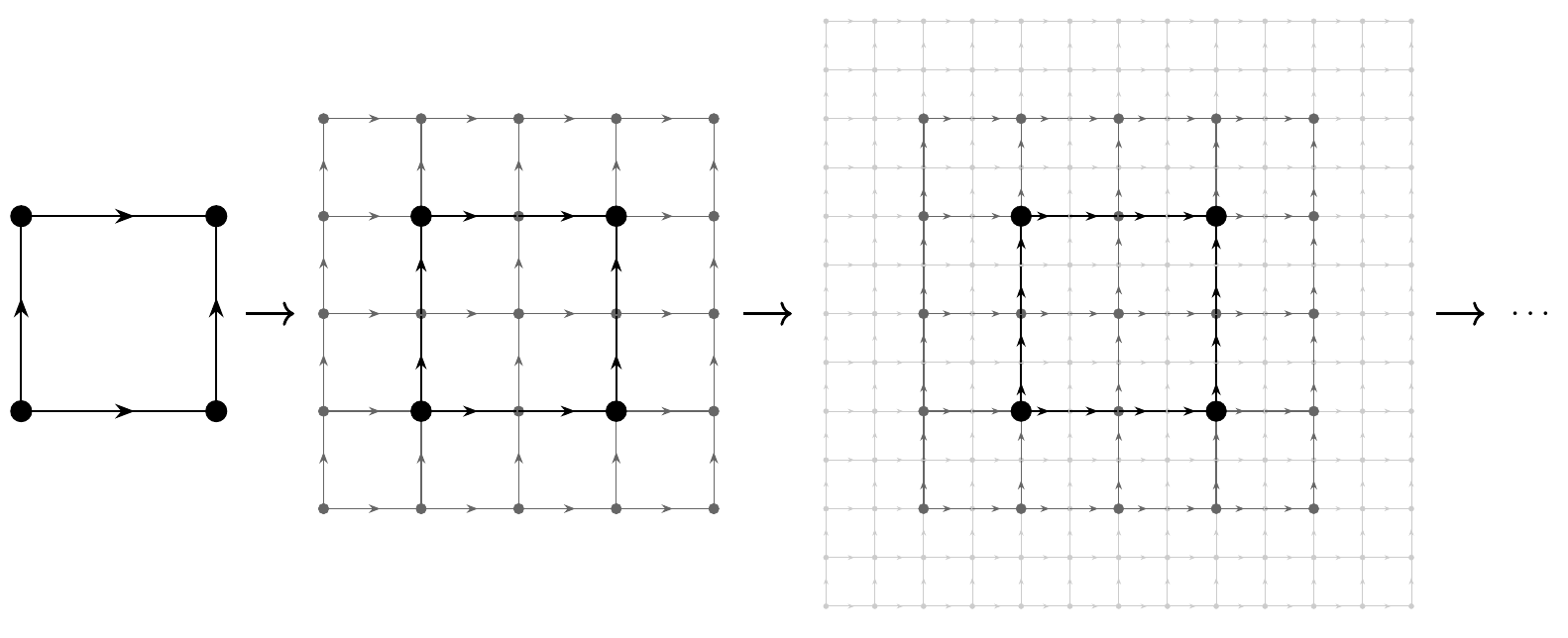}
\caption{The notion of refiniments of graphs as introduced above allows to simultaneously take the thermodynamic and the continuum limit of the system.}
\end{figure}

\noindent
Given a refinement $\iota: \Lambda_i \to \Lambda_j$ of two graphs $\Lambda_i,\Lambda_j$, we introduce a map $\R_{i,j} \colon \spc {j} \rightarrow \spc{i}$ between the corresponding configuration spaces as follows.
Given an edge $e \in \Lambda_i$, we let
\begin{align}
\R_{i,j}(a)_e = a_{e_1} \cdots a_{e_n},
\label{eq:R0}
\end{align}
where $\iota^{(1)}_{i,j}(e) = (e_1,\ldots,e_n)$.
The compatibility of these maps under composition is readily checked and we arrive at the following result.
\begin{prop}
\label{prop:refinements_of_conf}
There exists a contravariant functor from $\mathsf{Refine}$ to the category of compact Hausdorff spaces that sends a graph $\Lambda_i$ to the space $\spc{i}$, and a refinement $(\Lambda_i, \Lambda_j, \iota_{i,j})$ to the map $\R_{i,j}$..
\end{prop}

\begin{rem}
A particular consequence of the above proposition is that a direct system  $((\Lambda_i)_{i \in I}, (\iota_{i,j})_{i,j \in I, \: i \leq j})$ in $\mathsf{Refine}$ induces an inverse system\\
$((\spc{i})_{i \in I}, (\R_{i,j})_{i,j \in I, \: i \leq j})$ in the category of compact Hausdorff spaces.
In what follows, we will construct various other co- and contravariant functors from $\mathsf{Refine}$ to certain categories, which induce direct and inverse systems in these categories, respectively.
For the sake of brevity, we will write the above direct system as $(\Lambda_i, \iota_{i,j})$, and do the same with other direct and inverse systems.
\end{rem}

\subsubsection{Elementary refinements}
\label{sect:elementary_refinements}

In what follows we need to carry out a number of computations, some of which are rather tedious to write out for arbitrary refinements. We simplify our computations by making use of the fact that any refinement can be decomposed into the composition of \emph{elementary refinements}. This is in line with \cite[Lemma 4]{baez96}, although we do not admit the reversal of the orientation of an edge.
More precisely, given an arbitrary refinement $(\Lambda_i, \Lambda_j, \iota)$, there exists a sequence $(\Lambda_k, \Lambda_{k + 1}, \iota_{k,k + 1})_{k = 0}^{n - 1}$ of refinements such that $\Lambda_0 = \Lambda_i$, $\Lambda_n = \Lambda_j$, $\iota = \iota_{n - 1,n} \circ \dots \circ \iota_{0,1}$, and for each $i \in \{0,\ldots,n - 1\}$, the refinement $(\Lambda_k, \Lambda_{k + 1}, \iota_{k,k + 1})$ falls into one of the following two classes of examples:
\begin{itemize}
\item The graph $\Lambda_{k + 1}$ is obtained from $\Lambda_k$ by adding an extra edge
\begin{center}
\includegraphics[width=0.35\textwidth]{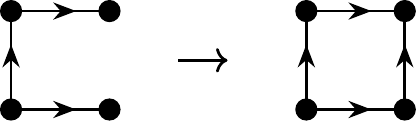}
\end{center}
or by adding an extra vertex and an extra edge:
\begin{center}
\includegraphics[width=0.35\textwidth]{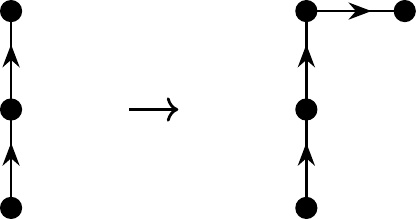}
\end{center}
\noindent
At the level of configuration spaces, both of these embeddings induce the map
\begin{equation}
\label{eq:addededge}
\R_{k, k + 1} \colon \spc {k + 1} \rightarrow \spc {k}, \quad 
((a_{e})_{e \in \Lambda^1_{k}}, a_{e_0}) \mapsto (a_e)_{e \in \Lambda^1_{k}},
\end{equation}
where $e_0 \in \Lambda^{(1)}_{k + 1}$ denotes the `added' edge.

\item The graph $\Lambda_{k + 1}$ is obtained from $\Lambda_k$ by subdividing an edge into two edges:
\begin{center}
\includegraphics[width=0.5\textwidth]{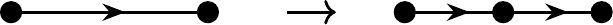}
\end{center}

\noindent
This type of embedding induces the following map between configuration spaces:
\begin{equation}
\label{eq:edgesubd}
\R_{k,k + 1} \colon \spc {k + 1} \rightarrow \spc {k}, \quad 
((a_{e})_{e \in \Lambda^1_{k} - \{e_0\}}, a_{e_1}, a_{e_2}) \mapsto ((a_e)_{e \in \Lambda^1_{k} - \{e_0\}}, a_{e_1} a_{e_2}),
\end{equation}
where $e_0 \in \Lambda^1_{k}$ denotes the edge that is `subdivided' into $e_1$ and $e_2$.
\end{itemize}

\noindent
It is clear that each embedding $(\Lambda_i, \Lambda_j,\iota_{i,j})$ can be decomposed into a sequence of elementary refinements of these two types, so we just check that the corresponding map $R_{i,j} : \spc{j} \to \spc{i}$ is independent of the chosen sequence $(\Lambda_k, \Lambda_{k + 1}, \iota_{k,k + 1})_{i = 0}^{n - 1}$.
This comes down to checking that given two pairs of elementary refinements \begin{align*}
&((\Lambda_k, \Lambda_{k + 1}, \iota_{k,k + 1}), (\Lambda_{k + 1}, \Lambda_{k + 2}, \iota_{k + 1,k + 2})), \\
&((\Lambda_k, \Lambda^\prime_{k + 1}, \iota_{k,k + 1}'), (\Lambda^\prime_{k + 1}, \Lambda_{k + 2}, \iota_{k + 1,k + 2}')),
\end{align*}
the corresponding maps between configuration spaces 
satisfy
\begin{equation}
\R_{k, k + 1} \circ \R_{k + 1, k + 2} = \Rp_{k, k + 1} \circ \Rp_{k + 1, k + 2}.
\label{eq:independence_of_elementary_refinements}
\end{equation}
We distinguish three possible cases:
\begin{itemize}
\item If all of the above refinements are additions of a single edge, with the second pair put in opposite order with respect to the first, then the induced maps between the configuration spaces can be commuted and Equation \eqref{eq:independence_of_elementary_refinements} holds;

\item Suppose that $\iota_{k,k + 1}$ is the addition of an edge, while $\iota_{k + 1,k + 2}$ subdivides that edge into two edges, and the refinements $\iota_{k,k + 1}',\iota_{k + 1,k + 2}'$ in the second pair are both additions of an edge:
\begin{center}
\includegraphics[width=0.9\textwidth]{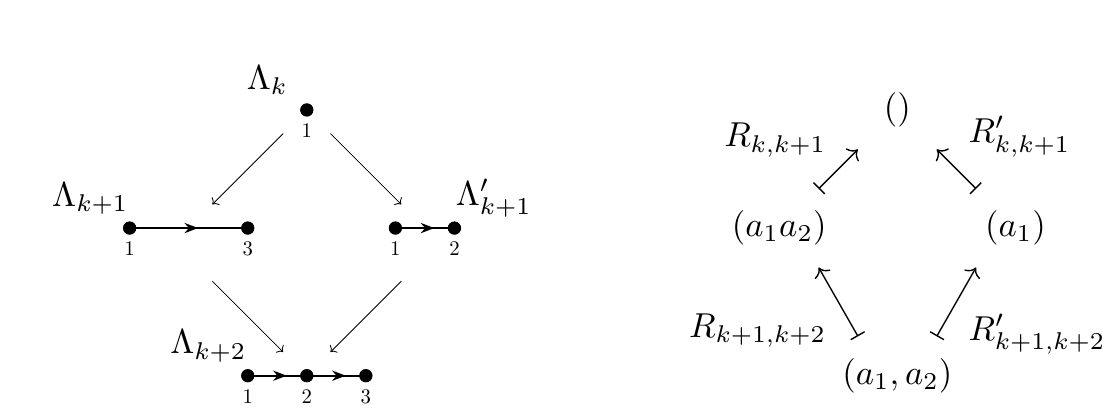}
\end{center}
The maps in the diagram on the right commute, which shows that Equation \eqref{eq:independence_of_elementary_refinements} holds;

\item Suppose that both pairs of refinements correspond to the subdivision of an edge into two edges but in different orders. Then we have

\begin{center}
\includegraphics[width=0.9\textwidth]{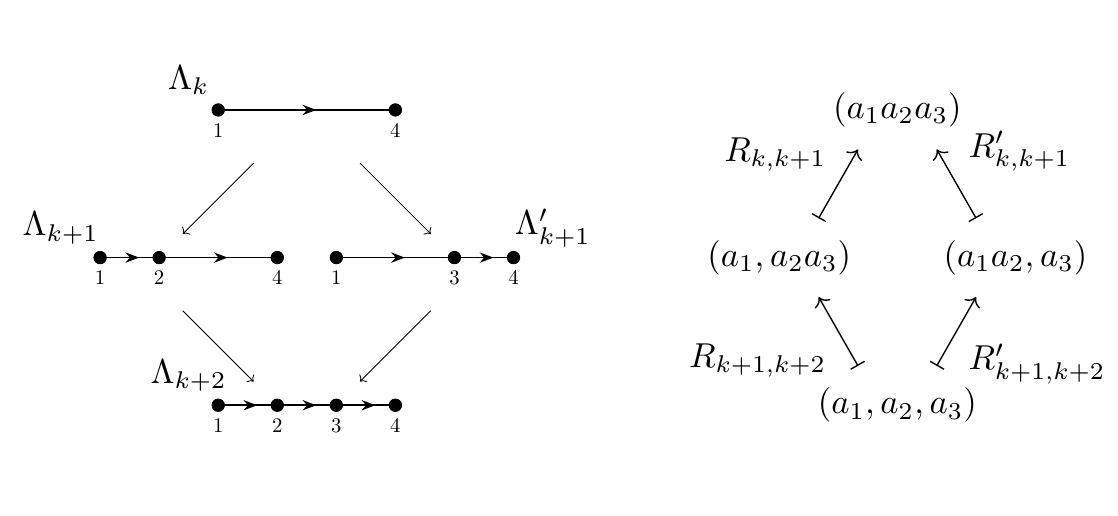}
\end{center}
from which we infer that also in this case, Equation \eqref{eq:independence_of_elementary_refinements} holds.
\end{itemize}
Throughout the rest of the text, whenever we discuss a refinement $(\Lambda_i, \Lambda_j, \iota)$ of graphs, we shall only discuss the cases in which the refinements are elementary refinements, and use the above observation to extend statements to the general case.

\subsubsection{The action of the gauge group}

Let us fix two graphs $\Lambda_i$ and $\Lambda_j$ together with a refinement $(\Lambda_i, \Lambda_j, \iota)$. 

The map $\iota^{(0)}$ induces a surjective group homomorphism between the gauge groups given by pull-back:
\begin{equation}
\label{eq:grpHom}
(\iota^{(0)})^\ast \colon \cG_j \rightarrow \cG_i, \quad 
g = (g_x)_{x \in \Lambda_j^0} \mapsto (g_{\iota^{(0)}(x)})_{x \in \Lambda_i^0},
\end{equation}
Clearly, this map can be directly factorized into products of maps corresponding to elementary refinements. 

Moreover, one readily verifies that $\R_{i,j} \colon \spc j \rightarrow \spc i$ satisfies the equivariance condition 
\begin{equation*}
(\iota^{(0)})^\ast(g) \cdot \R_{i,j} (a) = \R_{i,j} (g \cdot a),
\end{equation*} 
for all $g \in \cG_j$ and $a \in \spc j$, hence it descends to a map $\Rr_{i,j} \colon \cG_j \backslash \spc j \rightarrow \cG_i \backslash \spc i$. 

If we let $\pi_i \colon \spc i \rightarrow \cG_i \backslash \spc i$ denote the canonical projection, we obtain a commutative diagram:
\begin{equation}
\label{eq:commSpaces}
\xymatrix{
\spc {j}  \ar[d]_{\pi_j} \ar[r]^{\R_{i,j} }& \spc{i} \ar[d]_{\pi_i}\\
\cG_j \backslash \spc{j} \ar[r]^{\Rr_{i,j} } & \cG_i\backslash \spc{i} 
}
\end{equation}

\begin{prop}
There exists a contravariant functor from $\mathsf{Refine}$ to the category of compact Hausdorff spaces that sends a graph $\Lambda_i$ to the space $\cG_i \backslash \spc{i}$, and a refinement $(\Lambda_i, \Lambda_j, \iota_{i,j})$ to the map $\Rr_{i,j}$.
\end{prop}

\subsection{Hilbert spaces}

Next, we construct the Hilbert spaces of square integrable functions with respect to the Haar measure, and define the corresponding maps between them.
We start by recalling some results for the (Haar) measures on the configuration spaces, originally derived in \cite{ALMMT96,baez94,Lew94} (see also \cite{baez96,Ha94}). 
 
\begin{lem}
\label{lem:inv_sys_spaces}
On the inverse system of Hausdorff spaces $(\spc{i}, \R_{i,j})$ we have an \emph{exact} inverse system of measures for $(\spc{i}, \R_{i,j})$, i.e.\ a collection of Radon measures $\mu_i$ on $\spc{i}$ such that for $i \leq j$ one has $(\R_{i,j})_\ast(\mu_j)=\mu_i$.
In particular, the image of the Haar measure on $G^{\Lambda_j^1}$ under the map induced by the map $\R_{i,j}$ is the Haar measure on $ G^{\Lambda_i^1}$.
\end{lem}
\begin{proof}
The first part of the theorem follows from the Riesz--Markov representation theorem.

We will check the second part of the statement for the elementary refinements of Section \ref{sect:elementary_refinements}. 
Let $\phi$ be a continuous function on $G^{\Lambda_i^1}$ and let $\mu_{i+1}$ be the Haar measure on $G^{\Lambda_{i+1}^1}$.
By definition
\begin{equation*}
\int_{G^{\Lambda_i^1}} \phi \: \d (\R_{i,i+1})_\ast(\mu_{i+1}) := \int_{G^{\Lambda_{i+1}^1}} \left(\R_{i,i+1}\right)^*(\phi) \: \d\mu_{i+1} .
\end{equation*}
We will show that $(\R_{i,i+1})_\ast(\mu_{i+1})$ is left invariant, i.e. that 
\begin{equation*}
\int_{G^{\Lambda_i^1}} L_h \phi \: \d (\R_{i,i+1})_\ast(\mu_{i+1}) = \int_{G^{\Lambda_i^1}}\phi \: \d (\R_{i,i+1})_\ast(\mu_{i+1}) ,
\end{equation*}
for any $h \in G^{\Lambda^1_i}$.
Since the Haar measure on $G^{\Lambda_i^1}$ is the product of $\vert \Lambda_i^1 \vert$ Haar measures on $G$ it follows that
\begin{equation*}
\int_{G^{\Lambda_i^1}} L_h \phi \: \d (\R_{i,i+1})_\ast(\mu_{i+1})  
= \int_{G^{\Lambda_{i+1}^1}} \phi (h^{-1}\cdot \R_{i,i+1}(a)) \: \d\mu_{i+1}(a). \end{equation*}
An elementary refinement consisting of the addition of an edge amounts to forgetting an integration variable so there is nothing to prove. For the 

subdivision of an edge $e$ into $e_1 e_2$ 
we have
\begin{align*}
&\int_{G^{\Lambda_i^1}} \phi (h^{-1}\cdot \R_{i,i+1}(a)) \: \d\mu_{i+1}(a) \\
&= \int_{G^{\Lambda_i^1}} \phi \left( h^{-1}\cdot \left( (a_{e})_{e \in \Lambda^1_i \setminus\lbrace e_0 \rbrace}, a_{e_1} a_{e_2} \right) \right) \: \d\mu_{i+1}(a) \\ 
&= \int_{G^{\Lambda_i^1}} \phi \circ \R_{i,i+1} ((h^{-1} \cdot ((a_{e})_{e \in \Lambda^1_i \setminus\lbrace e_0 \rbrace}, a_{e_1})), a_{e_2} ) \: \d\mu_{i}((a_e)_{e \in \Lambda^1_i \setminus \lbrace e_0 \rbrace}, a_{e_1}) \mathord{\times} \d\mu(a_{e_2}) \\ 
&= \int_{G^{\Lambda_i^1}} \phi \circ \R_{i,i+1} \left( ( (a_{e})_{e \in \Lambda^1_i \setminus \lbrace e_0 \rbrace}, a_{e_1}), a_{e_2} \right) \: \d\mu_{i}((a_e)_{e \in \Lambda^1_i \setminus\lbrace e_0 \rbrace}, a_{e_1}) \times \d\mu(a_{e_2}) \\
&= \int_{G^{\Lambda_i^1}} \phi(\R_{i,i+1} (a)) \: \d\mu_{i+1}(a),
\end{align*}
where in the second last equality we have used left-invariance of the Haar measure $\d \mu_{i}$.
 \end{proof}

\begin{prop}
On the inverse system of Hausdorff spaces $(\cG_i \backslash \spc i, \Rr_{{i,j}})$ we have an \emph{exact} inverse system of measures.
\end{prop}

\begin{proof}
By the Riesz--Markov representation theorem, the projection $\pi_i: \spc i \to \cG_i \backslash \spc i$ induces a map from the space of Radon measures on $\spc i$ to the space of Radon measures on $ \cG_i \backslash \spc i$. Equation \eqref{eq:commSpaces} then implies the existence of a commutative diagram between the corresponding spaces of Radon measures.
\end{proof}

\noindent
Let us now dualize this construction on the measure spaces $\spc{i}$ and construct a direct system of Hilbert spaces $L^2(\spc{i})$. Let us write $R : \spc{j} \to \spc{i}$ for a map between configuration spaces induced by an arbitrary refinement $\iota: \Gamma_i \to \Gamma_j$. We then set
\begin{equation}
\label{eq:isometry}
u := (\R)^\ast \colon L^2(\spc{i}) \rightarrow L^2(\spc{j}), \quad 
\psi \mapsto \psi \circ \R.
\end{equation}
Moreover, we define
\begin{equation*}
u^\red := (\Rr)^\ast \colon L^2(\cG_i \backslash \spc{i}) \rightarrow L^2(\cG_j \backslash \spc{j}), \quad 
\psi \mapsto \psi \circ \Rr.
\end{equation*}

\begin{prop}
If $p_i$ is the map given by
\begin{equation*}
p_i \colon L^2(\spc{i}) \mapsto L^2(\cG_i \backslash \spc{i}), \quad 
\psi \mapsto \left( \cG_i a \mapsto \int_{\cG_i} \psi(g \cdot a) \: \d\mu_{\cG_i}(g) \right),
\end{equation*}
for all $i$, where $\mu_{\cG_i}$ denotes the Haar measure on $\cG_i$, then:
\begin{myenum}
\item The pullback $\pi_i^\ast: L^2(\cG_i \backslash \spc{i}) \to L^2(\spc{i})$ of $\pi_i$ is the adjoint of $p_i$ (for all $i$).

\item The following squares
\begin{equation*}
\xymatrix{
L^2(\spc{i})\ar[r]^u & L^2(\spc{j}) \\
L^2(\cG_i \backslash \spc{i}) \ar[u]^{\pi_i^\ast} \ar[r]^{u^\red} & L^2(\cG_j \backslash \spc{j})\ar[u]_{\pi_j^\ast} \\
}
\end{equation*}
and
\begin{equation*}
\xymatrix{
L^2(\spc{i})\ar[r]^u \ar[d]_{p_i} & L^2(\spc{j}) \ar[d]^{p_j}  \\
L^2(\cG_i \backslash \spc{i}) \ar[r]^{u^\red} & L^2(\cG_j \backslash \spc{j})\\
}
\end{equation*}
commute.

\item The maps $u$ and $u^\red$ are isometries.
\end{myenum}
\label{prop:refinement_on_Hilbert_spaces1}
\end{prop}

\begin{proof}
(1) For $\psi \in L^2(\spc{i}),\varphi \in L^2(\cG_i \backslash \spc{i})$ we have that
\begin{align*}
\inp{\varphi}{p_i \psi}_{L^2(\cG_i \backslash \spc{i})}
&= \int_{\cG_i \backslash \spc{i}} \overline{\varphi(\cG_i a)} \int_{\cG_i} \psi(g \cdot a) \: \d\mu_{\cG_i}(g) \: \d\mu_{\cG_i \backslash \spc{i}}(a) \\
&= \int_{\spc{i}} \overline{\varphi \circ \pi_i(a)} \int_{\cG_i} \psi(g \cdot a) \: \d\mu_{\cG_i}(g) \: \d\mu_i(a) \\
&= \int_{\spc{i}} \int_{\cG_i} \overline{\varphi \circ \pi_i(g^{-1} \cdot a)} \psi(a) \: \d\mu_{\cG_i}(g) \: \d\mu_i(a) \\
&=  \int_{\spc{i}} \overline{\varphi \circ \pi_i(a)} \psi(a) \: \d\mu_i(a)
= \inp{\pi_i^\ast(\varphi)}{\psi}_{L^2(\spc{i})},
\end{align*}
where we have used bi-invariance of the Haar measure on $\spc{i}$ in the fourth step.

Commutativity of the first square in (2) follows directly from the fact that $\pi_i \circ \R = \Rr \circ \pi_j$, which holds by definition of $\Rr$. For the second square to commute, we let $a \in \spc{j}, \psi \in L^2(\spc{i})$ and compute that indeed
\begin{align*}
(p_j \circ u(\psi))(\cG_j a)
&= \int_{\cG_j} (u(\psi))(g \cdot a) \: \d\mu_{\cG_j}(g) 
= \int_{\cG_j} \psi \circ \R (g \cdot a) \: \d\mu_{\cG_j}(g)\\
&= \int_{\cG_j} \psi((\iota^0)^\ast(g) \cdot \R (a)) \: \d\mu_{\cG_j}(g) \\
&= \int_{G^{\Lambda_j^0 - \iota^0(\Lambda_i^0)}} \int_{\cG_i} \psi(g \cdot \R (a)) \: \d\mu_{\cG_i}(g) \: \d \nu(g^\prime) \\
&= \int_{\cG_i} \psi(g \cdot \Rr(a)) \: \d\mu_{\cG_i}(g)  
=  p_i(\psi)(\cG_i \R(a)) \\
&= (u^\red \circ p_i(\psi))(\cG_j a),
\end{align*}
where $\nu$ denotes the Haar measure on $G^{\Lambda_j^0 - \iota^0(\Lambda_i^0)}$.

For (3) we use that by the very definition of the measures on the spaces $\cG_i \backslash \spc{i}$ and $\cG_j \backslash \spc{j}$, the maps $\pi_i^\ast$ and $\pi_j^\ast$ are isometries.
Thus, by commutativity of the first square in (2), it suffices to show that $u$ is an isometry.
We will prove the statement for the elementary refinements discussed in Section \ref{sect:elementary_refinements}. 
Let $\psi \in L^2(\spc{i})$.
\begin{itemize}
\item If $\Lambda_j$ is obtained from $\Lambda_i$ by adding an edge $e^\prime \in \Lambda_j^1$ then
\begin{align*}
\|u(\psi)\|_{L^2(\spc{j})}^2
&= \int_{G^{\Lambda_j^1}} |u(\psi)((a_e)_{e \in \Lambda_j^1})|^2 \: \d\mu_j((a_e)_{e \in \Lambda_j^1}) \\
&= \int_{G} \int_{G^{\Lambda_i^1}} |\psi((a_e)_{e \in \Lambda_i^1})|^2 \: \d\mu_i((a_e)_{e \in \Lambda_j^1}) \: \d\mu(a_{e^\prime}) \\
&= \int_{G^{\Lambda_i^1}} |\psi((a_e)_{e \in \Lambda_i^1})|^2 \: \d\mu_i((a_e)_{e \in \Lambda_i^1}) \\
&= \|\psi\|_{L^2(\spc{i})}^2,
\end{align*}

\item If $\Lambda_j$ is obtained from $\Lambda_i$ by subdividing an edge $e_0 \in \Lambda_i^1$ into two edges $e_1, e_2 \in \Lambda_j^1$, then
\begin{align*}
& \|u(\psi)\|_{L^2(\spc{j})}^2
= \int_{G^{\Lambda_j^1}} |u(\psi)((a_e)_{e \in \Lambda_j^1})|^2 \: \d\mu_j((a_e)_{e \in \Lambda_j^1}) \\
& \quad = \int_{G} \int_{G} \int_{G^{\Lambda_i^1 \setminus \{e_0\}}} |\psi((a_e)_{e \in \Lambda_i^1 \setminus \{e_0\}}, a_{e_1}a_{e_2})|^2 \: \d\nu((a_e)_{e \in \Lambda_i^1 \setminus \{e_0\}}) \: \d\mu(a_{e_1}) \d\mu(a_{e_2}) \\
& \quad = \int_{G} \int_{G^{\Lambda_i^1 \setminus \{e_0\}}} |\psi((a_e)_{e \in \Lambda_i^1 \setminus \{e_0\}}, a_{e_2})|^2 \: \d\nu((a_e)_{e \in \Lambda_i^1 \setminus \{e_0\}}) \: \d\mu(a_{e_2}) \\
& \quad = \int_{G^{\Lambda_i^1}} |\psi((a_e)_{e \in \Lambda_i^1})|^2 \: \d\mu_i((a_e)_{e \in \Lambda_i^1}) = \|\psi\|_{L^2(\spc{i})}^2,
\end{align*}
since the Haar measure is left-invariant and normalized.
Here, $\mu$ and $\nu$ denote the Haar measures on $G$ and $G^{\Lambda_i^1 \setminus \{e_0\}}$, respectively.\qedhere
\end{itemize}
\end{proof}

\begin{prop}
There exist two covariant functors from $\mathsf{Refine}$ to the category of Hilbert spaces that send a graph $\Lambda_i$ to the spaces $L^2(\spc{i})$ and $L^2(\cG_i \backslash \spc{i})$, and a refinement $(\Lambda_i, \Lambda_j, \iota_{i,j})$ to the linear isometries $u_{i,j}$ and $u^\red_{i,j}$, respectively.
\label{prop:Hilbert_space_refinements}
\end{prop}

\begin{proof}

Let $\Lambda_i$, $\Lambda_j$ and $\Lambda_k$ be three graphs, with corresponding spaces of connections $\spc{i}$, $\spc{j}$ and $\spc{k}$ and gauge groups $\cG_i$, $\cG_j$ and $\cG_k$.
Suppose in addition that we are given refinements $(\Lambda_i, \Lambda_j, \iota_{i,j})$ and $(\Lambda_j, \Lambda_k, \iota_{j,k})$.

We need to prove that the corresponding maps between Hilbert spaces and observable algebras satisfy
\begin{align*}
u_{i,k} &= u_{j,k} \circ u_{i,j}; \\
u^\red_{i,k} & = u^\red_{j,k} \circ u^\red_{i,j}; 
\end{align*}
The fact that $u_{i,k} = u_{j,k} \circ u_{i,j}$ follows from Remark \ref{rem:successive_refinement1} and the definition of the map $\R$.
To prove $u^\red_{i,k} = u^\red_{j,k} \circ u^\red_{i,j}$, note that for $\Lambda_i \leq \Lambda_j$, the maps $p_i^\ast$ and $p_j^\ast$ are isometries by definition of the measure on $\cG_i \backslash \spc{i}$ and $\cG_j \backslash \spc{j}$. Thus $p_i p_i^\ast = \Id_{L^2(\cG_i \backslash \spc{i})}$ and $p_j p_j^\ast = \Id_{L^2(\cG_j \backslash \spc{j})}$.

Commutativity of the first square in Proposition \ref{prop:refinement_on_Hilbert_spaces1} and the fact that $p_i^\ast$ and $p_j^\ast$ are sections of $p_i$ and $p_j$, respectively, imply that $p_i \circ u \circ p_i^\ast = u^\red$, and that $u$ maps $\cG_i$-invariant functions to $\cG_j$-invariant functions.
Observing that $p_i^\ast p_i$ and $p_j^\ast p_j$ are the orthogonal projections onto the spaces of $\cG_i$- and $\cG_j$-invariant functions, respectively, we infer that
\begin{align*}
u^\red_{i,k}
& = p_k \circ u_{i,k} \circ p_i^\ast
= p_k \circ u_{j,k} \circ u_{i,j} \circ p_i^\ast 
 = p_k \circ u_{j,k} \circ p_j^\ast p_j \circ u_{i,j} \circ p_i^\ast
= u^\red_{j,k} \circ u^\red_{i,j},
\end{align*}
which proves the claim.
\end{proof}

\subsection{Observable algebras}
The isometries between the Hilbert spaces constructed in the previous subsection naturally induce maps between the observables. In fact, we have the following:
\begin{prop}
The maps
\begin{equation*}
\begin{array}{lll}
v \colon B_0(L^2(\spc{i})) \rightarrow B_0(L^2(\spc{j})), \quad 
&b \mapsto u b u^\ast; \\
v^\red \colon B_0(L^2(\cG_i \backslash \spc{i})) \rightarrow B_0(L^2(\cG_j \backslash \spc{j})), \quad 
&b \mapsto u^\red b (u^\red)^\ast,
\end{array}
\end{equation*}
are injective $^\ast$-homomorphisms.
\label{prop:injective_star_homomorphisms}
\end{prop}

\begin{proof}
It is clear that $v$ and $v^\red$ respect its linear structures as well as the involutions.
Since $u$ and $u^\red$ are isometries, we have
\begin{equation*}
u^\ast u = \text{\normalfont Id}_{L^2(\spc{i})}, \quad  \text{\normalfont and} \quad (u^\red)^\ast u^\red = \text{\normalfont Id}_{L^2(\cG_i \backslash \spc{i})},
\end{equation*}
from which it readily follows that the maps $v$ and $v^\red$ are injective and respect the algebra structures.
\end{proof}

\noindent
Thus the maps $u$ and $v$ are embeddings of the `coarse' Hilbert space and observable algebra into the corresponding `finer' structures, respectively, and the maps $u^\red$ and $v^\red$ are their `reduced' counterparts.
We can now formulate the analogue of Proposition \ref{prop:refinement_on_Hilbert_spaces1} for the observable algebras:

\begin{prop}
Define the maps $P_i$, $P_j$, $\Pi_i$ and $\Pi_j$ by
\begin{equation*}
\begin{array}{ll}
P_i \colon B_0(L^2(\spc{i})) \rightarrow B_0(L^2(\cG_i \backslash \spc{i})), \quad 
& b \mapsto p_i b p_i^\ast; \\
P_j \colon B_0(L^2(\spc{j})) \rightarrow B_0(L^2(\cG_j \backslash \spc{j})), \quad 
& b \mapsto p_j b p_j^\ast; \\
\Pi_i \colon B_0(L^2(\cG_i \backslash \spc{i})) \rightarrow B_0(L^2(\spc{i})), \quad 
& b \mapsto p_i^\ast b p_i; \\
\Pi_j \colon B_0(L^2(\cG_j \backslash \spc{j})) \rightarrow B_0(L^2(\spc{j})), \quad 
& b \mapsto p_j^\ast b p_j.
\end{array}
\end{equation*}
Then the following squares

\[
\xymatrix{
B_0(L^2(\spc{i})) \ar[r]^{v} & B_0(L^2(\spc{j})) \\
B_0(L^2(\cG_i \backslash \spc{i})) \ar[r]^{v^\red} \ar[u]^{\Pi_i} & B_0(L^2(\cG_j \backslash \spc{j})) \ar[u]_{\Pi_j}
}
\]
and
\[
\xymatrix{
B_0(L^2(\spc{i})) \ar[d]_{P_i} \ar[r]^{v} & B_0(L^2(\spc{j})) \ar[d]^{P_j}\\
B_0(L^2(\cG_i \backslash \spc{i})) \ar[r]^{v^\red} & B_0(L^2(\cG_j \backslash \spc{j}))
}
\]
commute.

\label{prop:refinement_on_observable_algebras1}
\end{prop}

\begin{proof}
We shall only present a proof of commutativity of the first square; commutativity of the second square can be proved in a similar fashion.
Let $b \in B_0(L^2(\cG_i \backslash \spc{i}))$.
Then using the commutativity of the first square in Proposition \ref{prop:refinement_on_Hilbert_spaces1}, we obtain
\begin{align*}
v \circ \Pi_i(b)
& = u p_i^\ast b p_i u^\ast 
= (u p_i^\ast) b (u p_i^\ast)^\ast \\
& = (p_j^\ast u^\red) b (p_j^\ast u^\red)^\ast  
= p_j^\ast u^\red b (u^\red)^\ast p_j
= \Pi_j \circ v^\red(b),
\end{align*}
as desired.
\end{proof}

\begin{prop}
There exist two covariant functors from $\mathsf{Refine}$ to the category of \Cs algebras that send a graph $\Lambda_i$ to the spaces $B_0(L^2(\spc{i}))$ and $B_0(L^2(\cG_i \backslash \spc{i}))$, and a refinement $(\Lambda_i, \Lambda_j, \iota_{i,j})$ to the injective $^\ast$-homomorphisms $v_{i,j}$ and $v^\red_{i,j}$, respectively.
The collections $(B_0(L^2(\spc{i})), v_{i,j})$ and $(B_0(L^2(\cG_i \backslash \spc{i})),v^\red_{i,j})$ form direct systems of \Cs algebras.
\label{prop:dir-syst-compacts}
\end{prop}

\begin{proof}
This follows from the fact that
\begin{equation*}
v_{i,k} = v_{j,k} \circ v_{i,j}; \quad
v^\red_{i,k} = v^\red_{j,k} \circ v^\red_{i,j},
\end{equation*}
which is a direct consequence of Proposition \ref{prop:Hilbert_space_refinements}.
\end{proof}

\noindent
We are also interested in describing the refinements of the observable algebras in purely geometric terms, that is to say, in terms of the pair groupoids $\sfG_i = \spc{i} \times \spc{i}$ that we associated to a graph $\Lambda_i$.
A map $R_{i,j} \colon \spc{j} \to \spc{i}$ canonically gives rise to a groupoid morphism $\sfR_{i,j} = \left( \sfR_{i,j}^{(0)}, \sfR_{i,j}^{(1)} \right) \colon \sfG_{j} \to \sfG_{i}$, where $\sfR_{i,j}^{(0)} = R_{i,j}$ and $\sfR_{i,j}^{(1)} = R_{i,j} \times R_{i,j}$.
Similarly, we obtain a groupoid morphism $\sfRr_{i,j} \colon \sfG^\red_j \rightarrow \sfG^\red_i$ between the pair groupoids associated to the reduced configuration spaces.
The following proposition is then an immediate consequence of Proposition \ref{prop:refinements_of_conf}:

\begin{prop}
There exist contravariant functors from $\mathsf{Refine}$ to the category of groupoids that send a graph $\Lambda_i$ to the groupoids $\sfG_i$ and $\sfG^\red_i$, and a refinement $(\Lambda_i, \Lambda_j, \iota_{i,j})$ to the groupoid morphisms $\sfR_{i,j}$ and $\sfRr_{i,j}$.
\label{prop:refinement_of_groupoids1}
\end{prop}

\noindent 
More interestingly, the maps $\sfR_{i,j}$ induce a map $\sfR_{i,j}^\ast$ between the groupoid \Cs algebras $C^\ast(\sfG_i)$ and $C^\ast(\sfG_j)$, given simply by pull-back.
We will show that it coincides with $v_{i,j}=v$ from Proposition \ref{prop:injective_star_homomorphisms}, after identifying $C^\ast(\sfG_i) \simeq B_0(L^2(\spc{i}))$, using the isomorphism induced by the map defined in Equation \eqref{eq:integral-operator}.

\begin{prop}
The following diagram
\begin{equation}
\xymatrix{
C^\ast(\sfG_i)\ar[d]_{\simeq} \ar[r]_{\sfR_{i,j}^\ast} & C^*(\sfG_j) \ar[d]^{\simeq} \\
B_0(L^2(\spc{i})) \ar[r]_{v_{i,j}} & B_0(L^2(\spc{j}))
}
\end{equation}
commutes.
\end{prop}

\begin{proof}
With $u_{i,j}: L^2(\spc{i}) \to L^2(\spc{j})$ as defined in Equation \eqref{eq:isometry}, we have to establish that
\begin{equation*}
u_{i,j} (T_h) u_{i,j}^* = T_{R_{i,j}^\ast(h)}.
\end{equation*}
We will prove this for the elementary refinements of Section \ref{sect:elementary_refinements} where we assume for simplicity that $\Lambda_i$ is a graph with one edge $e$, and that $\Lambda_j$ is a graph with two edges $e_1$ and $e_2$ such that $(e_1,e_2)$ is a path in $\Lambda_j$. 
Let $\psi \in L^2(\spc{j}), \varphi \in L^2(\spc{i})$.
\begin{itemize}
\item If $\Lambda_j$ is obtained from $\Lambda_i$ by adding the edge $e_2 \in \Lambda_i^1$ 
 then we have
\begin{align*}
\inp{\varphi}{u^\ast \psi}_{L^2(\spc{i})}
&= \inp{u \varphi}{\psi}_{L^2(\spc{j})} \\
&= \int_{G^2} \overline{u(\varphi)(a_{e_1}, a_{e_2})} \psi(a_{e_1}, a_{e_2}) \: \d\mu_{G^2}(a_{e_1}, a_{e_2}) \\
&= \int_G \int_G \overline{\varphi(a_{e_1})} \psi(a_{e_1}, a_{e_2}) \: \d\mu(a_{e_2}) \: \d\mu(a_{e_1}) \\
&= \int_{G^{\Lambda_i^1}} \overline{\varphi(a_e)} \int_G \psi(a_e, a_{e_2}) \: \d\mu(a_{e_2}) \: \d\mu(a_e),
\end{align*}
so
\begin{equation*}
u^\ast \psi(a_e)
= \int_G \psi(a_e, a_{e_2}) \: \d\mu(a_{e_2}).
\end{equation*}
Next, let $h \in C(\sfG_i^{(1)})$, and let $(a_{e_1}, a_{e_2}) \in \spc{j}$.
It follows that
\begin{align*}
&v(T_h)(\psi)(a_{e_1}, a_{e_2})
= u T_h u^\ast(\psi)(a_{e_1}, a_{e_2})
= T_h u^\ast(\psi)(a_{e_1}) \\
&\quad = \int_{G^{\Lambda_i^1}} h(b_{e_1}, a_{e_1}) u^\ast(\psi)(b_{e_1}) \: \d\mu_i(b_{e_1}) \\
&\quad = \int_G \int_G h(b_{e_1}, a_{e_1}) \psi(b_{e_1}, b_{e_2}) \: \d\mu(b_{e_1}) \: \d\mu(b_{e_2}) \\
&\quad = \int_{G^{\Lambda_j^1}} h(b_{e_1}, a_{e_1}) \psi(b_{e_1}, b_{e_2}) \: \d\mu_j(b_{e_1}, b_{e_2}).
\end{align*}
Thus, if we define the function $\bar{h} \in C(\sfG_j^{(1)})$ by
\begin{equation*}
\bar{h}((b_{e_1}, b_{e_2}), (a_{e_1}, a_{e_2}))
= h(b_{e_1}, a_{e_1}),
\end{equation*}
and define the corresponding integral operator $T_{\bar{h}} \in B_0(L^2(\spc{}))$, then we obtain $v(T_h) = T_{\bar{h}}$. It is not difficult to see that $\bar{h}$ is the composition of $h$ with the function 
$\sfR^{(1)} \colon \sfG_j^{(1)} \rightarrow \sfG_i^{(1)}$ given in this case by
\begin{equation*}
\sfR^{(1)} \left((b_{e_1}, b_{e_2}), (a_{e_1}, a_{e_2}) \right)= (\R \times \R)((b_{e_1}, b_{e_2}), (a_{e_1}, a_{e_2}))= (b_{e_1}, a_{e_1})
\end{equation*}

\item Now suppose that $\Lambda_j$ is obtained from $\Lambda_i$ by subdividing the edge $e \in \Lambda_i^1$ into the two edges $e_1$ and $e_2 \in \Lambda_j^1$.
Then
\begin{align*}
\inp{\varphi}{u^\ast \psi}_{L^2(\spc{i})}
&= \inp{u \varphi}{\psi}_{L^2(\spc{j})} \\
&= \int_{G^2} \overline{u(\varphi)(a_{e_1}, a_{e_2})} \psi(a_{e_1}, a_{e_2}) \: \d\mu_{G^2}(a_{e_1}, a_{e_2}) \\
&= \int_G \int_G \overline{\varphi(a_{e_1} a_{e_2})} \psi(a_{e_1}, a_{e_2}) \: \d\mu(a_{e_2}) \: \d\mu(a_{e_1}) \\
&= \int_G \overline{\varphi(a_{e_2})} \int_G \psi(a_{e_1}, a_{e_1}^{-1} a_{e_2}) \: \d\mu(a_{e_1}) \: \d\mu(a_{e_2}),
\end{align*}
so
\begin{equation*}
u^\ast \psi(a_e)
= \int_G \psi(a^\prime, (a^\prime)^{-1} a_e) \: \d\mu(a^\prime).
\end{equation*}
Next, let $h \in C(\sfG_i^{(1)})$, and let $(a_{e_1}, a_{e_2}) \in \sfG_j^{(0)}$.
Left-invariance of the Haar measure now implies
\begin{align*}
&v(T_h)(\psi)(a_{e_1}, a_{e_2})
= u T_h u^\ast(\psi)(a_{e_1}, a_{e_2})
= T_h u^\ast(\psi)(a_{e_1} a_{e_2}) \\
&\quad = \int_G h(b_{e_2}, a_{e_1} a_{e_2}) u^\ast(\psi)(b_{e_2}) \: \d\mu(b_{e_2}) \\
&\quad = \int_G \int_G h(b_{e_2}, a_{e_1} a_{e_2}) \psi(b_{e_1}, b_{e_1}^{-1} b_{e_2}) \: \d\mu(b_{e_1}) \: \d\mu(b_{e_2}) \\
&\quad = \int_G \int_G h(b_{e_1} b_{e_2}, a_{e_1} a_{e_2}) \psi(b_{e_1}, b_{e_2}) \: \d\mu(b_{e_1}) \: \d\mu(b_{e_2}) \\
&\quad = \int_{G^{\Lambda_j^1}} h(b_{e_1} b_{e_2}, a_{e_1} a_{e_2}) \psi(b_{e_1}, b_{e_2}) \: \d\mu_j(b_{e_1}, b_{e_2}).
\end{align*}
Thus, if we define the function $\bar{h} \in C(\sfG_j^{(1)})$ by
\begin{equation*}
\bar{h}((b_{e_1}, b_{e_2}), (a_{e_1}, a_{e_2}))
= h(b_{e_1} b_{e_2}, a_{e_1} a_{e_2}),
\end{equation*}
and define the corresponding integral operator $T_{\bar{h}} \in B_0(L^2(\spc{}))$, then we obtain $v(T_h) = T_{\bar{h}}$.
Again, the map $\bar{h}$ is the composition of $h$ with the function $\sfR^{(1)} \colon \sfG_j^{(1)} \rightarrow \sfG_i^{(1)}$ which in this case given by 
\begin{equation*}
\sfR^{(1)}((b_{e_1}, b_{e_2}), (a_{e_1}, a_{e_2}))
= (R \times \R)((b_{e_1}, b_{e_2}), (a_{e_1}, a_{e_2}))=(b_{e_1} b_{e_2}, a_{e_1} a_{e_2}).
\end{equation*}
\end{itemize}
This completes the proof.
\end{proof}

\noindent 
The statement can readily be modified for the groupoid \Cs algebras of the reduced groupoids.
We summarize the results obtained in this subsection in the following:

\begin{thrm}
\label{prop:refinement_of_grpd_alg}
The collections $\left( C^\ast(\sfG_i), \sfR_{i,j}^\ast \right)$
and $\left(C^\ast(\cG_i \backslash \sfG_i), {\sfRr_{i,j}}^\ast \right)$
with connecting maps induced by the maps $R_{i,j}$ and $\Rr_{i,j}$, respectively, form direct systems of \Cs algebras.
Moreover, these direct systems of \Cs algebras are isomorphic to the direct systems described in Proposition \ref{prop:dir-syst-compacts}.
\end{thrm}

\subsection{The Hamiltonian}
\label{subsec:The_Hamiltonian}

\noindent
Suppose again that we have fixed two graphs $\Lambda_i$, $\Lambda_j$ together with a refinement $(\Lambda_i, \Lambda_j, \iota)$.
Consider the Hamiltonians
\begin{equation*}
H_{0,i} = \sum_{e \in \Lambda_i^1} -\frac{1}{2} I_{i,e} \Delta_e, \quad \text{ and } \quad
H_{0,j} = \sum_{e \in \Lambda_j^1} -\frac{1}{2} I_{j,e} \Delta_e,
\end{equation*}
let $\cH_i := L^2(\spc{i})$, let $\cH_j := L^2(\spc{j})$ and let $u \colon \cH_i \rightarrow \cH_j$ and $u_\red \colon \cH_i^{\cG_i} \rightarrow \cH_j^{\cG_j}$ be the maps between the corresponding Hilbert spaces.

\begin{prop}
\label{prop:red_Hamiltonian}
Suppose that for each $e \in \Lambda_i^1$, we have
\begin{equation}
I_{i,e} = \sum_{k = 1}^n I_{j,e_k},
\label{eq:path_lengths}
\end{equation}
where $\iota^{(1)}(e) = (e_1,e_2,\ldots,e_n)$.
\begin{myenum}
\item We have $u(\dom(H_{i,0})) \subseteq \dom(H_{j,0})$, and the following diagram
\begin{equation*}
\xymatrix{
\dom(H_{0,i}) \ar[r]^{u} \ar[d]_{H_{0,i}} & \dom(H_{0,j}) \ar[d]^{H_{0,j}} \\
\cH_i \ar[r]^{u} & \cH_j 
}
\end{equation*}
is commutative.

\item We have $u(\dom(H_{0,i}) \cap \cH_i^{\cG_i}) \subseteq \dom(H_{0,j}) \cap \cH_j^{\cG_j}$, and the following diagram
\begin{equation*}
\xymatrix{
\dom(H_{0,i}) \cap \cH_i^{\cG_i} \ar[d]_{H^\red_{0,i}}\ar[r]^{u_\red} & \dom(H_{0,j}) \cap \cH_j^{\cG_j} \ar[d]^{H^\red_{0,j}} \\
\cH_i^{\cG_i} \ar[r]^{u_\red} & \cH_j^{\cG_j} 
}
\end{equation*}
is commutative, where $H^\red_{0,i}$ denotes the restriction of $H_{0,i}$ to $\dom(H_{0,i}) \cap \cH_i^{\cG_i}$, and $H^\red_{0,j}$ is defined analogously.
\end{myenum}
\end{prop}
\proof
(1) As before, we shall provide a proof of the proposition for the elementary refinements discussed in Section \ref{sect:elementary_refinements}, and for the sake of simplicity, we assume that $\Lambda_i$ is the graph consisting of one edge $e$.
It is clear that $u(C^\infty(\spc{i})) \subseteq C^\infty(\spc{j})$.
Now let $\psi \in C^\infty(\spc{i})$, let $X \in \fg$, and let $(a_1, a_2) \in \sfG^{(0)}_j$.
\begin{itemize}
\item If $\Lambda_j$ is obtained from $\Lambda_i$ by adding the edge $e_2 \in \Lambda_j^1$ then $I_{j,(s(e_1),t(e_1))} = I_{i,(s(e),t(e))}$ and we have trivially that
$$
H_{j,0}(u(\psi))(a_1, a_2) = -\frac{1}{2} I_{j,(s(e_1),t(e_1))}  \Delta_{e_1} \psi(a_{1}) = 
\left( u \circ H_{i,0} (\psi) \right)(a_1,a_2). 
$$
\item If $\Lambda_j$ is obtained from $\Lambda_i$ by subdividing the edge $e \in \Lambda_i^1$ into the two edges $e_1$ and $e_2 \in \Lambda_j^1$  then $I_{j,(s(e_1),t(e_1))} +  I_{j,(s(e_2),t(e_2))} = I_{i,(s(e),t(e))}$ and 
\begin{align*}
&H_{0,j}(u(\psi))(a_1, a_2) \\
&\quad = -\frac{1}{2} \left( I_{j,(s(e_1),t(e_1))}\Delta_{e_1} (u(\psi))(a_{1}, a_{2}) + I_{j,(s(e_2),t(e_2))} \Delta_{e_2}( u(\psi)) (a_{1}, a_{2}) \right) \\
&\quad = -\frac{1}{2} \left( I_{j,(s(e_1),t(e_1))} +  I_{j,(s(e_2),t(e_2))} \right) \Delta_{e}
(\psi)(a_{1} a_{2}) \\
&\quad = -\frac{1}{2} I_{i,(s(e),t(e))} (u \circ \Delta_e(\psi))(a_1, a_2) \\
&\quad = (u \circ H_{0,i}(\psi))(a_1, a_2).
\end{align*}
using invariance of the Laplacian on $L^2(G)$ with respect to the left and right action of $G$ in going to the third line.

\end{itemize}
This proves commutativity of the diagram for the restrictions of the operators to the spaces of smooth functions.
The assertion now follows from the fact that $u$ is a bounded operator and the fact that $H_{0,i}$ and $H_{0,j}$ are the closures of their restrictions to $C^\infty(\spc{i})$ and $C^\infty(\spc{j})$, respectively.

(2) The inclusion is a consequence of the first part of this proposition, and the definition of $u_\red$.
Now let $p_i := p_{\cH_i^{\cG_i}}$, let $p_j := p_{\cH_j^{\cG_j}}$, and consider the following cube:
\[
\xymatrix{
& \dom(H_{0,j})\ar[rr]^{H_{0,j}}\ar'[d][dd]^(.45){p_j}& & \cH_j \ar[dd]^{p_j}\\
\dom(H_{0,i})\ar[dd]^{p_i}\ar[rr]^{H_{0,i}}\ar[ur]^{u} & & \cH_i \ar[ur]^{u} \ar[dd]^(.65){p_i}&  \\
& \dom(H_{0,j}) \cap \cH_j^{\cG_j} \ar'[r][rr]^(-.45){H_{0,j}^\red} & & \cH_j^{\cG_j} \\
\dom(H_{0,i}) \cap \cH_i^{\cG_i} \ar[rr]^{H_{0,i}^\red} \ar[ur]_{u^\red} & & \cH_i^{\cG_i} \ar[ur]_{u^\red} & 
}
\]

The top face is commutative by the previous part of the proposition.
The side faces of the cube are commutative by Proposition \ref{prop:refinement_on_Hilbert_spaces1}.
The front and rear faces of the cube are commutative by Proposition \ref{prop:reduction_of_the_electric_part_of_the_Hamiltonian}, and by the same proposition, the map $p_i \colon \dom(H_{0,i}) \rightarrow \dom(H_{0,i}) \cap \cH_i^{\cG_i}$ is surjective.
It follows that the bottom face of the cube is commutative, which is what we wanted to show.
\endproof

\begin{rem}
Because the edges of the graphs under consideration correspond to paths in space, and because we subdivide these paths into smaller paths when considering finer graphs, we can take the constant $I_e$ to be proportional to the length of the path associated to the edge $e$ for each $e \in \Lambda^1$.
In this way, Equation \eqref{eq:path_lengths} will be satisfied.
This generalizes the electric part of the Kogut-Susskind Hamiltonian found in \cite{kogut75}, which is proportional to the lattice spacing, to finite graphs.
\end{rem}

 \section{The continuum limit}
\label{sec:limit}

\noindent
We will now consider the continuum limit of our theory by considering the limit objects of the inverse and direct systems constructed in the previous section. This includes inverse limits of measure spaces and groupoids, and the direct limits of Hilbert spaces and (groupoid) \Cs algebras. In particular, we will identify a limit pair groupoid $\sfG_\infty$ for which the groupoid \Cs algebra $C^\ast(\sfG_\infty)$ is isomorphic to the limit of observable algebras $\varinjlim_{i \in I} A_i$.

First of all, the inverse system $(\spc{i}, \R_{i,j})$ in Lemma \ref{lem:inv_sys_spaces} has a limit in the category of topological spaces, which is unique up to unique isomorphism, and which can be realised as follows:
\begin{equation*}
\spc{\infty} = \varprojlim_{i \in I} \spc{i} :=  \left\lbrace  a=(a_i)_{i \in I} \in \prod_{i \in I} \spc{i} \ \vert \ a_i = \R_{i,j} (a_j) \ \mbox{for all} \ i \leq j  \right\rbrace 
\end{equation*}
together with maps
\[\R_{i,\infty}: \spc{\infty} \to \spc{i},\]
which are given by the projection. Note that since the maps $\R_{i,j}$ are not group homomorphism, the limit space $\spc{\infty}$ does not automatically possess a group structure.

By \cite[Lemma 1.1.10]{RZ2010}, since $\spc{\infty}$ is an inverse limit of compact Hausdorff spaces, the maps $\R_{i,\infty}$ are surjective for all $i \in I$. Moreover, since the spaces involved are compact, the maps $\R_{i,j}$ are automatically proper and so are the structure maps $\R_{i,\infty}$.
The existence of a measure on the limit space is then a consequence of Prokhorov's theorem (\cite[Theorem 21]{S73}): 

\begin{prop}
\label{prop:meas_lim_sp}
Let $\spc{\infty}$ denote the limit of the inverse system of measurable topological spaces $((\spc{i}, \mu_i), \R_{i,j})$. Then there exists a Radon measure ${\mu}_{\infty} $ on $\spc{\infty}$ such that $\R_{i,\infty} ({\mu}_{\infty})=\mu_i$.
\end{prop}
\noindent
By Proposition \ref{prop:Hilbert_space_refinements} we have a direct system of Hilbert spaces $(\cH_i, u_{i,j})$, where $\cH_i := L^2(\spc{i},\mu_i)$. Its direct limit is nothing but the space of $L^2$ functions on the inverse limit of the spaces of connections with respect to the inverse limit measure (cf. \cite{baez96}):
\begin{equation*}
{\cH}_{\infty}:= \varinjlim_{i \in I} \cH_i \simeq L^2(\spc{\infty}, {\mu}_{\infty}).
\end{equation*}
The following proposition, which relates the inverse limit of Hilbert spaces with the direct limit of their algebras of observables, is probably well-known to experts. For the sake of completeness, we include a proof. 
\begin{prop}[]
Let $((\cH_i, \inp{\cdot}{\cdot}_i), u_{i,j})$ be a direct system of Hilbert spaces such that each map $u_{i,j}$ is an isometry.
Let $(\cH_\infty ,\inp{\cdot}{\cdot})$ be its direct limit.
For each $i,j \in I$ with $i \leq j$, define the map $v_{i,j}$ by
\begin{equation*}
v_{i,j} \colon B_0(\cH_i) \rightarrow B_0(\cH_j), \quad 
a \mapsto u_{i,j} a u_{i,j}^\ast.
\end{equation*}
Then $v_{i,j}$ is injective $^\ast$-homomorphism, hence it is an isometry.
Furthermore, $(B_0(\cH_i), v_{i,j})$ is a direct system of \Cs algebras, and we have
\begin{equation*}
\varinjlim_{i \in I} B_0(\cH_i) \simeq B_0(\cH_\infty).
\end{equation*}
\label{prop:direct_limit_of_Hilbert_spaces_and_compact_operators}
\end{prop}
\begin{proof}
Since each $u_{i,j}$ is an isometry, each $\cH_i$ is embedded into $\cH_j$ for $i \leq j$ and it is further embedded in the direct limit $\cH$ via the maps $u_{i,\infty}$ that appear in the definition of direct limit.
 
The maps $u_{i,\infty}$ induce isometric *-homomorphisms $B_0(\cH_i) \to B_0(\cH)$ given by
$T \mapsto u_{i,\infty} T u_{i, \infty}^*$. Hence we have that the direct limit of the algebras of compact operators, which by injectivity of the structure maps $v_{i,j}$ is the closure of the union of the algebras, satisfies 
\begin{equation*}
\varinjlim_{i\in I} B_0(\cH_i)
= \overline{\bigcup_{i\in I} B_0(\cH_i)}
\subseteq B_0(\varinjlim_{i\in I}\cH_i).
\end{equation*}
To prove the reverse inclusion, recall that by definition $B_0(\cH_\infty)$ is the closure of the linear span of rank 1 operators on $\cH_\infty$, which we write as
$\theta_{\xi,\eta}$ with $\|\xi\|,\|\eta\| \leq 1$.
By definition of the direct limit we have that $\cH_\infty = \overline{\bigcup_{i\in I}\cH_i}$.
So for every $0 <\epsilon<1$ and $i \in I$ we can find $\xi', \eta' \in u_{i,\infty}(\cH_i)$ with $\| \xi -\xi'\| \leq \epsilon/3$ and $\| \xi -\xi'\| \leq \epsilon/3$, which implies that $\| \theta_{\xi,\eta}-\theta_{\xi',\eta'} \| \leq \epsilon$.
This means that the closed set $\varinjlim_{i \in I} B_0(\cH_i)$ contains all finite rank operators on $\cH$, hence it contains the whole of $B_0(\cH_\infty)$.
\end{proof}

\noindent
In our case of interest, we have:
\begin{cor}
\label{cor:dir-limit-observable}
The direct limit of the observable algebras is given by
\begin{equation*}
\varinjlim_{i\in I} A_i \simeq B_0(L^2(\spc{\infty})).
\end{equation*}
\end{cor}

Next, we determine the inverse limit of the groupoids $\sfG_i$ and show that the direct limit \Cs algebra $\rmA_{\infty} = B_0(L^2(\spc{\infty}, \mu_{\infty}))$ agrees with the \Cs algebra $C^*(\sfG_{\infty})$ of the inverse limit groupoid $\sfG_{\infty}$. 

Given the simple structure of the groupoid morphisms $\sfR_{i,j}: \sfG_j \to \sfG_i$ one easily checks that the limit groupoid $\sfG_{\infty}$ is also a pair groupoid and is given by
\begin{equation*}
\sfG_{\infty} = \spc{\infty} \times \spc{\infty}.
\end{equation*}
It is by definition a free and transitive groupoid.

\begin{rem}
More generally, the limit of an inverse family of compact \emph{transitive} groupoids such that all groupoid homomorphisms are surjective is also transitive.
Moreover, for inverse families of compact \emph{free} groupoids, the limit is also a free groupoid.
The proofs rely on the fact that 
source and target in the limit groupoid are defined component-wise.
\end{rem}

\noindent
On the groupoid $\sfG_{\infty} = \spc{\infty} \times \spc{\infty}$ we have a natural Haar system given by 
\[\lbrace \mu \times \delta_x \ \vert \ x \in \spc{\infty} \rbrace,\] where $\delta_x$ is the unit point mass at $x$ and $\mu$ is a positive Radon measure on $\spc{\infty}$ of full support. 

\begin{thrm}
\label{thm:limits-reduced}
The groupoid \Cs algebra $C^*(\spc{\infty} \times \spc{\infty})$ is isomorphic to the limit observable algebra $A_\infty$, which in turn is isomorphic to $B_0(L^2(\spc{\infty},\mu))$, where $\mu$ is the injective limit of the measures on $\spc{i}$.
\end{thrm}
\proof
Since the limit measure $\mu_{\infty}$ on the space $\spc{\infty}$ is a positive Radon measure of full support, the result follows from the second isomorphism in Equation \eqref{eq:groupoid-C*-compacts} and Corollary \ref{cor:dir-limit-observable} above.
\endproof

\begin{rem}The question whether the \Cs algebras associated with two Haar system on a given groupoid are isomorphic was answered positively by Muhly, Renault and Williams for the case of transitive groupoids (cf. \cite[Theorem 3.1]{MRW87}), of which pair groupoids are a special case.
Hence the choice of Haar system does not affect, in our setting, the structure of the groupoid \Cs algebra.
For a more in depth discussion on the dependence of the groupoid \Cs algebra on the choice of Haar system we refer the reader to \cite[Section~5]{buneci} and \cite[Section~3.1]{Pat99}.
\end{rem}

\noindent
Thus we have
\begin{equation*}
C^\ast(\sfG_\infty) = C^{\ast}(\varprojlim_{i\in I} \sfG_i) 
\simeq B_0(L^2(\spc{\infty},\mu))
\simeq \varinjlim_{i \in I} B_0(L^2(\spc{i}, \mu_i))
\simeq \varinjlim_{i \in I} C^\ast(\sfG_i),
\end{equation*}
justifying the idea that the quantized algebra of observables on the inverse limit $\spc{\infty}$ is the direct limit of the quantized algebras of observables on the spaces $(\spc{i})_{i \in I}$.
Intuitively, we may interpret the groupoid \Cs algebra $C^\ast(\sfG_\infty)$ as the quantization of the infinite-dimensional phase space $T^*\spc{\infty}$ (which we will not attempt to define in this paper).

Let us spend a few words on the free Hamiltonian in the continuum limit.
In fact, since the sequence of Hamiltonians $H_{0,i}$ on $\cH_i$ is compatible (in the sense of Proposition \ref{prop:red_Hamiltonian}) with the direct system of Hilbert spaces, it is not difficult to show that there is a limit operator $H_{0,\infty}$ on $\cH_\infty$ that is self-adjoint on a suitable domain $\dom(H_{0,\infty})$.
The spectral decomposition of this operator can then be shown to be well-behaved with respect to the spectral decompositions of each $H_{0,i}$.
In contrast with the spectral properties of each $H_{0,i}$, it is less clear what the summability properties of such an operator are, as for instance infinite multiplicities will appear.
We leave the analysis of the Hamiltonian in the limit for future work. 

\begin{rem}
The dynamics of the above quantum lattice gauge theory (and including fermions) in the (thermodynamic) limit is also the subject of \cite{GR15}.
There, the authors do not study the limit of the Hamiltonians directly, but rather focus on the limit of the one-parameter subgroups generated by the (interacting) Hamiltonians at all finite levels. 
\end{rem}

\subsection{Quantum gauge symmetries and the continuum limit}
We finish this paper by discussing the reduction of the quantum system in the limit. 

By equivariance of the maps involved in the refinement procedure as described in Equation \eqref{eq:commSpaces}, the results of Proposition \ref{prop:meas_lim_sp} holds \emph{verbatim} for the inverse family of quotient measure spaces with respect to the action of the gauge group. If we let $\spc{\infty}^{\red}$ denote the limit of the inverse system of topological measure spaces $((\cG_i \backslash \spc{i}, \mu^{\red}_i), (\Rr_{i,j})_{i,j \in I})$, then there exists a Radon measure ${\mu}^{\red}_{\infty} $ on $\spc{\infty}^{\red}$ such that $\Rr_{i,\infty} ({\mu}_{\infty}^{\red})=\mu_i^{\red}$. 

Next, we can consider the space of square integrable functions on $\spc{\infty}^{\red}$ with respect to the Radon measure ${\mu}^{\red}_{\infty}$. Then the direct limit of the direct system of Hilbert spaces $(L^2(\cG_j \backslash \spc{j}),u^\red_{j,k})$ of Proposition \ref{prop:Hilbert_space_refinements} is given by
\[\cH^{\red}_{\infty} \simeq  L^2(\spc{\infty}^\red, \mu_{\infty}^{\red}).\]
An application of Proposition \ref{prop:direct_limit_of_Hilbert_spaces_and_compact_operators} yields 
\begin{equation*}
\varinjlim_{i \in I} B_0(\cH_i) \simeq B_0(\cH_{\infty}^{\red}).
\end{equation*}
As in the previous section we may then infer that the underlying groupoid for the observable algebra of the reduced quantum system is a direct limit of pair groupoids so that 
\begin{equation*}
C^*(\sfG_{\infty}^{\red})\simeq B_0(L^2(\spc{\infty}^\red,\mu^\red_{\infty})) .
\end{equation*}
In other words, we have arrived at the reduced analogue of Theorem \ref{thm:limits-reduced}.

\subsection{Outlook}

\noindent 
For the quantization of the configuration space we have followed the approach of \cite{landsman95} and defined the quantized algebra of observables as a groupoid \Cs algebra. The merit of this approach is that it is fully compatible with the natural maps between configuration spaces induced by graph refinements. Hence it allowed us to concretely describe the observable algebras in both the continuum and the themodynamic limit. 
Moreover, in the case of the thermodynamic limit, our algebra of observables agrees with the infinite tensor product algebra $\mathcal{L}$ considered in \cite{grundling13} (where each projection $P_k$ is the orthogonal projection onto the space of constant functions on $G$).

However, when we want to extend the above kinematical description of the limiting quantum gauge system to incorporate the Hamiltonian dynamics for the interacting system, we run into the following problems.
Namely, since our limit observable algebra is given by the space of compact operators, it does not really capture the infinite number of degrees of freedom that one would expect for an interacting quantum field theory (cf. \cite{Yng04} for a nice overview of this point), or in the description of the statistical physics of an infinite system at finite temperature \cite{AW63}.
As such, our limit observable algebra only admits KMS-states that are associated to inner automorphisms of the algebra, which prompts the question whether it is the right algebra for the description of a nontrivial quantum field theory.

The reason for this lack of interesting states might be that even though our choice of maps between configuration spaces is natural from a \emph{classical} point of view, the induced maps 
$v_{i,j}$ between the different observable algebras defined in Proposition \ref{prop:refinement_on_observable_algebras1} do not induce maps between the  state spaces of the algebras.


It is in this context interesting to mention that there are other approaches to the construction of the limit observable algebra, one of which was developed by Kijowski in \cite{kijowski77}, and later by Oko\l{}\'ow in \cite{okolow13} (cf. \cite{kijowski16}), and recently explored in depth by Lan\'ery and Thiemann in a series of papers \cite{lanery14-I,lanery14-II,lanery14-III,lanery15-IV}, see \cite{lanery16} for a comprehensive overview of these papers.
The main point where their approach differs from ours, is that they assume of the existence of a canonical unitary map between Hilbert spaces, which they use to define injective $^\ast$-homomorphism between the corresponding algebras of bounded operators, and which ensures that the transpose of this homomorphism maps states to states, i.e. preserves the normalization of positive functionals.
However, in their approach, the maps at the level of bounded operators do not reduce to maps between the algebras of compact operators, thus forcing them to abandon the setting of \Cs algebraic quantization described in e.g. \cite{landsman98}.

Yet another approach is suggested by Grundling and Rudolph in \cite{GR15}. For the thermodynamic limit of lattice QCD, they identify as the observable algebra a \Cs algebra that is larger than the above kinematical \Cs algebra, but that is closed under a global time evolution generated by the (local) Hamiltonians. It would be interesting to see whether such a time evolution exists for the continuum limit as well, but we leave this and other questions for future research.

\end{document}